\documentclass{gen-j-l}%
\usepackage{amsfonts}
\usepackage{amsmath}
\usepackage{amssymb}
\usepackage{graphicx}%
\newtheorem{theorem}{Theorem}[section]

\theoremstyle{definition}

\theoremstyle{remark}
\newtheorem{remark}[theorem]{Remark}
\numberwithin{equation}{section}
\theoremstyle{plain}

\copyrightinfo{2001}{enter name of copyright holder}
\begin{document}
\title[Quotient of Normal Random Variables]{The Quotient of Normal Random Variables And Application to Asset Price Fat Tails}
\author{Carey Caginalp}
\address{Department of Mathematical Sciences, Carnegie Mellon University, Pittsburgh
PA\ 15213 and Mathematics Department, University of Pittsburgh, Pittsburgh, PA\ 15260}
\email{carey\_caginalp@alumni.brown.edu}
\urladdr{http://www.pitt.edu/\symbol{126}careycag/}
\author{Gunduz Caginalp}
\address{Mathematics Department, University of Pittsburgh, Pittsburgh, PA\ 15260}
\email{caginalp@pitt.edu}
\urladdr{http://www.pitt.edu/\symbol{126}caginalp/}
\date{September 28, 2017}
\keywords{Quotient of Normals, Heavy Tails, Fat Tails, Supply/Demand, Tail of
Distribution, Leptokurtosis, Asset Price Change, Returns on Stocks}

\begin{abstract}
The quotient of random variables with normal distributions is examined and
proven to have have power law decay, with density $f\left(  x\right)  \simeq
f_{0}x^{-2}$, with the coefficient depending on the means and variances of the
numerator and denominator and their correlation. We also obtain the
conditional probability densities for each of the four quadrants given by the
signs of the numerator and denominator for arbitrary correlation $\rho
\in\lbrack-1,1).$ For $\rho=-1$ we obtain a particularly simple closed form
solution for all $x\in$ $\mathbb{R}$. The results are applied to a basic issue
in economics and finance, namely the density of relative price changes.
Classical finance stipulates a normal distribution of relative price changes,
though empirical studies suggest a power law at the tail end. By considering
the supply and demand in a basic price change model, we prove that the
relative price change has density that decays with an $x^{-2}$ power law.
Various parameter limits are established.

\end{abstract}
\maketitle

\section{Introduction}

A long-standing puzzle in economics and finance has been the "fat tails"
phenomenon in relative asset price change and other observations that refer to
the empirically observed power-law decay rather than the expected classical
exponential. In practical terms, this means unusual events are less rare than
expected, resulting in broad implications as discussed below. Given the quite
general application of the Central Limit Theorem, it is natural to expect that
the frequency of the relative price change as a function of the relative price
change would result in a normal, or Gaussian, distribution with exponential
decay, i.e., the density has the same tail as $f\left(  x\right)  =\left(
2\pi\sigma^{2}\right)  ^{-1/2}e^{-\frac{1}{2}\left(  x-\mu\right)  ^{2}%
/\sigma^{2}}$ where $\mu$ is the mean and $\sigma^{2}$ is the variance.

The assumption of a normal distribution for relative price changes dates back
to Bachelier's thesis \cite{BA} in 1900, and was further popularized in the
Black-Scholes work on options pricing \cite{BS}. Champagnat et. al. \cite{CH}
note several reasons for the saliency of utilizing log-normal price changes.
First, they can be "simply interpreted and estimated. Second, closed-form
expression exists for several options. Third, they could be embedded in a
continuous time process, as the geometric Brownian motion, which models the
evolution of the stock over the time. Indeed, many theories, for example the
Capital Asset Pricing Model (CAPM) for portfolio management, take their roots
in the Gaussian world."

One particular application involves "Value-at-Risk" which addresses questions
such as: Can we be sure that the investment will retain at least, say, $75\%$
of its current value within a five year time period with a $95\%$ confidence?
This methodology is often at the heart of risk analysis of an investment
portfolio. The Gaussian assumption facilitates calculations; however, there
have been numerous studies that indicate the risk is understated as a result
\cite{CH}, \cite{XA}, \cite{XI}. An early study by Fama \cite{F} provided
empirical evidence that there were ten times as many observations of relative
price changes than would be expected at four standard deviations from the
classical theories stipulating normal distributions. Mandelbrot and Hudson
\cite{MH} express their perspective in a section heading: "Markets are very,
very risky -- more risky than the standard theories imagine." In a more
generalized context, Taleb \cite{NG} has long asserted that unusual events
occur far more often than one would expect from the normal distribution.
Another aspect of work in this area has involved modeling \cite{LA}, \cite{M}.
The empirical observations of fat tails have also been noted in high frequency
trading \cite{DA}.

Thus, the question of the tail of the distribution of relative price changes
is important in several key areas of finance and economics. As a practical
matter it would appear that one can investigate empirically, and implement the
conclusions without reference to a particular model. While data seems to be
abundant, the problem is that obtaining a large amount of it often entails
using data from older time periods that may well be irrelevant. Hence, the
theoretical examination of the origins of the tail of the distribution becomes
crucial. Various explanations have been offered for the observation of fat
tails. These differ from our approach in that they stipulate an exogenous
influence that alters the natural normal decay. For example, it has been
argued that large institutional investors placing trades in less \ liquid
markets tend to cause large spikes \cite{GGP}. Theoretical models, e.g.,
\cite{M}, using random walk have been used to explain fat tails. We refer to
the book by Kemp \cite{KE} for a review of the literature. While these may be
important factors that lead to fat tails, we demonstrate below that fat tails
also arise from the endogenous price formation process.

In the classical approach to finance the basic starting point is the
stochastic differential equation for the relative price change, $P^{-1}%
dP=d\log P$ as a function of time, $t,$ and $\omega\in\Omega$ (with $\Omega$
as the sample space):
\begin{equation}
d\log P=\mu dt+\sigma dW. \label{dpdt}%
\end{equation}
Here $W$ is the standard Brownian motion, so $\Delta W:=W\left(  t\right)
-W\left(  t-\Delta t\right)  \sim\mathcal{N}\left(  0,\Delta t\right)  ,$
i.e., $W$ is normal with the variance as $\Delta t$ and mean $0$, and has
independent increments. The stochastic differential equation above is
shorthand for the integral form (suppressing $\omega$ in notation)
\begin{equation}
\log P\left(  t_{2}\right)  -\log P\left(  t_{1}\right)  =\int_{t_{1}}^{t_{2}%
}\mu dt+\int_{t_{1}}^{t_{2}}\sigma dW
\end{equation}
where $\mu$ and $\sigma$ can be constant, functions of $t$, or random
variables. For $\mu,\sigma$ constant, and $\Delta t:=t_{2}-t_{1},$ one can
write%
\begin{equation}
\Delta\log P:=\log P\left(  t_{2}\right)  -\log P\left(  t_{1}\right)
=\mu\Delta t+\sigma\Delta W.
\end{equation}
With the assumption that $\sigma$ is nearly constant over time, classical
finance clearly stipulates that the relative temporal changes in asset prices
should be normal. The basic equation (\ref{dpdt}) is obtained from the idea
that all information about the asset is incorporated into the price, and that
random world events alter the value on each time interval. Further, the
assumption of the existence of a vast arbitrage capital means that the changes
in the valuation are immediately reflected in the price, notwithstanding any
bias or mistake on the part of the less knowledgeable investors. Consistent
with the Central Limit Theorem, it is assumed that the events that alter the
asset valuation are normally distributed. But the important and tacit second
assumption is that relative price changes inherit this property.

While (\ref{dpdt}) is the basis for a large majority of papers on asset prices
and related issues in finance, it is difficult to generalize in some
directions. Indeed, with the assumption of infinite arbitrage already built
into the model, what can one subtract from infinity in order to obtain the
randomness that arises from the finiteness of trader assets and order flow?
What is needed then is an approach that takes into account more of the
microstructure of trading, i.e., the supply and demand of the asset submitted
to the market clearinghouse (see e.g., \cite{CB}, \cite{MA1},\cite{PP} and
references therein).

In this paper we examine the temporal evolution in percentage price changes by
modeling price change that utilizes a fundamental approach of supply/demand
economics analysis. We show that if one assumes supply and demand are normally
distributed, the mathematics of the quotient of normals suffices to yield fat
tails. In particular, in the limit of large deviations from the mean, one
obtains the result $f\left(  x\right)  \simeq f_{0}x^{-2}$ for the density for
large $x,$ where the constant $f_{0}$ depends on the means, variances and
correlations of supply and demand.

Thus modeling of the relative price change in terms of finite supply and
demand lead naturally to the basic statistical problem of the distribution of
the quotient of two normal random variables. While there is a long history of
the problem, surveyed below, most results concern the mid-range of the
distribution, rather than the tail. We obtain a number of exact
representations and rigorous bounds on the density conditioned upon the signs
of the numerator and denominator, as well as the overall density for all
correlations $\rho$ such that $\left\vert \rho\right\vert <1.$ For $\rho=-1$
we obtain a particularly simple exact expression for the density for all
$x\in\mathbb{R}.$

\section{The Model and the Quotient of Normals}

We focus on modeling price change in economics; the issues are similar in
other disciplines in which quotients arise. Classical economics stipulates
that prices remain constant in time when supply and demand are equal
\cite{WG}, \cite{HGH}, thereby defining equilibrium. When demand exceeds
supply, prices rise to restore equilibrium, and analogously in the other
direction. There is theoretical \cite{CB}, experimental \cite{CP} and
empirical \cite{CD} evidence that asset price change can be modeled as basic
goods, with the supply and demand depending on a number of factors such as the
cash/asset ratios.

Let $S\left(  t;\omega\right)  $ and $D\left(  t;\omega\right)  $ be the
supply and demand, respectively. The most general expression for relative
price change (see \cite{CB} and references, and, \cite{HQ} p. 165 for
motivation for the linear equation below) can be written as
\begin{equation}
\tau\frac{1}{P}\frac{dP}{dt}=g\left(  \frac{D-S}{S}\right)  ,
\end{equation}
where $\tau$ is the time scale, $P\left(  t;\omega\right)  $ is the unit price
of the asset, and $g$ is a differentiable function such that $g\left(
0\right)  =0,$ reflecting the assumption that prices do not change when demand
equals supply, and $g^{\prime}\left(  0\right)  =c>0,$ for some positive
constant $c,$ which stipulates that prices rise when demand exceeds supply. If
trading is very active, prices will react to small changes in supply/demand
imbalances, and the function $g$ can be assumed to be linear.\ Also, the
constants $\tau$ and $c$ can be incorporated into a dimensionless time,
yielding our basic starting point:%

\begin{equation}
\frac{1}{P\left(  t;\omega\right)  }\frac{d}{dt}P\left(  t;\omega\right)
=\frac{D\left(  t;\omega\right)  }{S\left(  t;\omega\right)  }-1.
\label{supplyDemandEqn}%
\end{equation}

Rather than considering the randomness directly in terms of prices, as in the
classical approach, we assume that supply and demand are random variables. In
a liquid market (i.e., frequently traded), the randomly flowing orders for
supply and demand can be approximated by normal distributions (as discussed
below). However, supply and demand are not likely to be independent, as a
random event that increases buying is likely to diminish selling. In general,
we can assume that $D$ and $S$ are bivariate normal random variables. This
leads to the question of estimating the tail of a distribution of
(\ref{supplyDemandEqn}) above, i.e., the quotient of bivariate normal random
variables. Our particular interest is for negatively correlated $D$ and $S$
but the analysis below will be for the full range of correlations, as similar
issues arise in other problems, e.g., physical, biological, in which there is
a quotient or random variables.

The supply and demand on a given interval consists of buy and sell orders
submitted to the market with some random distribution. The orders will be
influenced by news which will arrive from many independent sources, so that
the Central Limit Theorem will apply. That is, under a broad set of
conditions, the arrival of random orders, and thus, supply and demand, will be
approximated well by the normal distribution. However, price formation evolves
through a process that is almost deterministic. In other words, if one had a
large sample of the same supply/demand graphs, the resulting relative price
change would exhibit only a small variance as market makers and short term
traders can readily see the very short term market direction. This is a basic
consequence of economic game theory (see \cite{WG}, \cite{PP} and references
therein). In more practical terms, if we consider the stock of a major company
that trades with high volume, there will be a large number of market makers
whose only business is to profit from any deviations from the "correct" price,
given the order flow. Given a particular supply/demand graph, a significant
deviation can only arise from, not one, but many of these market makers erring
in the same direction. In summary, the randomness is inherently in the supply
and demand curves. For fixed supply/demand functions, the price evolution is
nearly deterministic.

The empirical assumption that relative price changes are normal (often stated
as price change is log-normal) has been tested, with results that indicate
significant deviations from normality (see e.g., \cite{F}, \cite{CH},
\cite{XA}). While there are numerous studies on the distribution of trading
prices, empirical studies testing the normality of buy/sell orders for major
stocks and commodities would be instrumental in understanding the problem of
fat-tails in relative price change.

We assume conditions on the orders that are compatible with the Central Limit
Theorem, so that, for $n$ large, the supply and demand are governed by a
bivariate normal distribution. Let $X_{1}:=D,\ X_{2}:=S$ and
\begin{equation}
R:=X_{1}/X_{2}-1.
\end{equation}
We assume that $\vec{X}:=\left(  X_{1},X_{2}\right)  ^{T}$ constitutes a
bivariate normal distribution having density, for $\vec{s}\in\mathbb{R}^{2},$
\begin{equation}
f\left(  \vec{s};\vec{\mu},\Sigma\right)  :=\left(  2\pi\right)
^{-1}\left\vert \Sigma\right\vert ^{-1/2}e^{-\frac{1}{2}\left(  \vec{s}%
-\vec{\mu}\right)  ^{T}\Sigma^{-1}\left(  \vec{s}-\vec{\mu}\right)  }\ ,
\end{equation}
where $\vec{\mu}:=\left(  \mu_{1},\mu_{2}\right)  $ and $\Sigma$ is the
covariance matrix,
\begin{equation}
\Sigma:=\left(
\begin{array}
[c]{cc}%
\sigma_{1}^{2} & \sigma_{12}\\
\sigma_{12} & \sigma_{2}^{2}%
\end{array}
\right)  \ \label{covMatrixDef}%
\end{equation}
with $\sigma_{12}:=E\left[  X_{1}X_{2}\right]  $ and $\sigma_{i}^{2}=E\left[
X_{i}^{2}\right]  $ for $i=1,2.$ One has the basic bound $\left\vert
\sigma_{12}\right\vert <\sigma_{1}\sigma_{2}$. Upon defining $\rho
:=\sigma_{12}/\left(  \sigma_{1}\sigma_{2}\right)  $ one has the result that
$\Sigma^{-1}$ exists if and only if $\left\vert \rho\right\vert <1$. Note that
$\vec{X}$ is said to have a singular bivariate normal distribution if there
exists real numbers $\mu_{a},\ \mu_{b},\ \sigma_{a},\ \sigma_{b}$ such that
$\vec{X}$ and $\left(  \sigma_{a}Y+\mu_{a},\sigma_{b}Y+\mu_{b}\right)  ^{T}$
are identically distributed, with $Y\sim\mathcal{N}\left(  0,1\right)  .$ This
is equivalent to $\left\vert \sum\right\vert =0$. We will consider the
$\left\vert \rho\right\vert =1$ case later and assume for now that
(\ref{covMatrixDef}) is nonsingular. An excellent source for relations on
multivariate normal distributions is Tong \cite{TO}.

We let $f_{X_{1}/X_{2}}\left(  x\right)  $ be the density, $F_{X_{1}/X_{2}%
}\left(  x\right)  ,$ the (cumulative) distribution function for the variable
$R$. While the random variables $X_{1}$ and $X_{2}$ can take on any values in
$\mathbb{R}$, the primary interest is in positive values which we consider
below, with similar results for the remaining quadrants of $\mathbb{R}^{2}$
presented in Appendix B.

\bigskip

The issue of the quotient of normal variables has been studied in a number of
contexts, as it arises in a number of biological and physical problems
including constructing the genome mapping of plants, imaging ventilation with
inert flourinated gases, and various neurological applications \cite{PA}. A
classical problem is to estimate parameters (e.g., mean and variance) of the
ratio of two populations. An early result by Geary \cite{GE} concerned the
ratio of two independent normals with zero means. Hinkley \cite{HI, HI2}
obtained a result in the limit as the variance of the denominator approached
zero. Kuethe \cite{KU} and Marsaglia \cite{MA} developed complex expressions
for the density of the ratio of two independent normals with strictly positive
means and variances.

Diaz-Frances and Rubio \cite{DR} obtained an important result by proving a
theorem that establishes bounds on the difference between the true
distribution and the proposed approximation. They summarize a number of
results and earlier works, noting that the ratio of independent normals with
positive means has no finite moments, and that the shape of the distribution
of the quotient "can be bimodal, asymmetric, symmetric, and even close to a
normal distribution, depending largely on the values of the coefficient of
variation of the denominator."

Much of the recent focus has been on approximating the ratio of the means by a
normal distribution; see Diaz-Frances and Rubio \cite{DR} and Diaz-Frances and
Sprott \cite{DS} for a discussion and references, and also Watson \cite{WA},
Palomino et al. \cite{PA}, Schneeweiss \cite{SC}, and Chamberlin and Sprott
\cite{CS}. For the most part the results involve the mid-range, where a normal
approximation is possible, rather than the tail.

The main focus of our paper will be on the tail of the distribution. Although
the terminology is not yet standard, one can broadly define "heavy tails" or
leptokurtosis as distributions with falloff less rapid than the normal, while
"fat tails" consist of power-laws.

While there has been some evidence that the ratio of independent normals,
under some conditions on the parameters, will be fat-tailed, there there has
not been a comprehensive understanding and analysis of the behavior of the
tail of the distribution. In Section 4 we prove that the ratio of two normals,
$R:=X_{1}/X_{2}$ with arbitrary means, variances and correlations $\rho
\in\lbrack-1,1)$ has the fat-tail property, with the density, $f_{X_{1}/X_{2}%
}\left(  x\right)  $ approaching zero as $x^{-2}.$ We also prove bounds on the
coefficient for large $\left\vert x\right\vert $, providing a rigorous
description of the tail of the quotient of normals under very general conditions.

\section{Calculation of the probability density of $X_{1}/X_{2}$, conditioned
on positive $X_{1}$\ and $X_{2}$\ .}

We assume $\mu_{1},\mu_{2},\sigma_{1}$ and $\sigma_{2}$ are all strictly
positive. Define
\begin{equation}
\phi\left(  z\right)  :=\left(  2\pi\right)  ^{-1/2}e^{-z^{2}/2}%
,\ \ \Phi\left(  z\right)  :=\int_{-\infty}^{z}\left(  2\pi\right)
^{-1/2}e^{-u^{2}/2}du.
\end{equation}
We note that the following probabilities are identical for $x>0$ (and zero for
$x\leq0$)%

\begin{gather}
\mathbb{P}\left\{  X_{1},X_{2}>0\ and\ \frac{X_{1}}{X_{2}}\leq x\right\}
=\mathbb{P}\left\{  X_{1},X_{2}>0\ and\ X_{1}\leq xX_{2}\right\} \nonumber\\
=\int_{0}^{\infty}ds_{2}\int_{0}^{xs_{2}}ds_{1}f\left(  s_{1},s_{2},\mu
_{1},\mu_{2},\Sigma\right) \nonumber\\
=\int_{0}^{\infty}ds_{2}\int_{0}^{xs_{2}}ds_{1}\frac{1}{\sigma_{1}\sigma
_{2}\sqrt{1-\rho^{2}}}\phi\left(  \frac{s_{2}-\mu_{2}}{\sigma_{2}}\right)
\phi\left(  \frac{s_{1}-\left(  \mu_{1}+\frac{\rho\sigma_{1}}{\sigma_{2}%
}\left(  s_{2}-\mu_{2}\right)  \right)  }{\sigma_{1}\sqrt{1-\rho^{2}}}\right)
\end{gather}
Various expressions for the density of bivariate and multivariate normals,
such as the one above, can be found in \cite{TO}.

In order to obtain the conditional density we use the identity%
\begin{equation}
\frac{d}{dx}\int_{0}^{xs_{2}}ds_{1}g\left(  s_{1},s_{2}\right)  =s_{2}g\left(
xs_{2},s_{2}\right)
\end{equation}
together with an implication of the Dominated Convergence Theorem to write%
\begin{align}
&  \frac{d}{dx}\mathbb{P}\left\{  X_{1},X_{2}>0\ and\ \frac{X_{1}}{X_{2}}\leq
x\right\} \nonumber\\
&  =\int_{0}^{\infty}s_{2}ds_{2}\frac{1}{\sigma_{1}\sigma_{2}\sqrt{1-\rho^{2}%
}}\phi\left(  \frac{s_{2}-\mu_{2}}{\sigma_{2}}\right)  \phi\left(
\frac{xs_{2}-\left(  \mu_{1}+\frac{\rho\sigma_{1}}{\sigma_{2}}\left(
s_{2}-\mu_{2}\right)  \right)  }{\sigma_{1}\sqrt{1-\rho^{2}}}\right)
\nonumber\\
&  =\frac{\left(  2\pi\right)  ^{-1}}{\sigma_{1}\sigma_{2}\sqrt{1-\rho^{2}}%
}\int_{0}^{\infty}s_{2}ds_{2}e^{-\frac{1}{2}\left(  \frac{s_{2}-\mu_{2}%
}{\sigma_{2}}\right)  ^{2}}e^{-\frac{1}{2}\left(  \frac{xs_{2}-\left(  \mu
_{1}+\frac{\rho\sigma_{1}}{\sigma_{2}}\left(  s_{2}-\mu_{2}\right)  \right)
}{\sigma_{1}\sqrt{1-\rho^{2}}}\right)  ^{2}}. \label{Q1half}%
\end{align}

Defining the quantities%

\begin{equation}
Q\left(  s\right)  :=\left(  \frac{s-\mu_{2}}{\sigma_{2}}\right)  ^{2}+\left(
\frac{xs-\left(  \mu_{1}+\frac{\rho\sigma_{1}}{\sigma_{2}}\left(  s-\mu
_{2}\right)  \right)  }{\sigma_{1}\sqrt{1-\rho^{2}}}\right)  ^{2}=As^{2}+2Bs+C
\end{equation}

\begin{align}
A  &  :=\sigma_{2}^{-2}+\sigma_{1}^{-2}\left(  1-\rho^{2}\right)  ^{-1}\left(
x-\rho\frac{\sigma_{1}}{\sigma_{2}}\right)  ^{2},\nonumber\\
B  &  :=-\sigma_{2}^{-2}\mu_{2}+\sigma_{1}^{-2}\left(  1-\rho^{2}\right)
^{-1}\left(  x-\rho\frac{\sigma_{1}}{\sigma_{2}}\right)  \left(  \mu_{2}%
\rho\frac{\sigma_{1}}{\sigma_{2}}-\mu_{1}\right)  ,\nonumber\\
C  &  :=\frac{\mu_{2}^{2}}{\sigma_{2}^{2}}+\left(  1-\rho^{2}\right)
^{-1}\left(  \rho\frac{\mu_{2}}{\sigma_{2}}-\frac{\mu_{1}}{\sigma_{1}}\right)
^{2}.
\end{align}%
\begin{equation}
Q\left(  s\right)  :=A\left(  s+\frac{B}{A}\right)  ^{2}-\frac{B^{2}}{A}+C
\end{equation}
we can express this equation as%
\begin{align}
\frac{d}{dx}\mathbb{P}\left\{  ...\right\}   &  =\frac{\left(  2\pi\right)
^{-1}}{\sigma_{1}\sigma_{2}\sqrt{1-\rho^{2}}}\int_{0}^{\infty}sdse^{-\frac
{1}{2}Q\left(  s\right)  }\nonumber\\
&  =\frac{\left(  2\pi\right)  ^{-1}}{\sigma_{1}\sigma_{2}\sqrt{1-\rho^{2}}%
}e^{\frac{1}{2}\left(  \frac{B^{2}}{A}-C\right)  }\int_{0}^{\infty}%
se^{-\frac{1}{2}A\left(  s+\frac{B}{A}\right)  ^{2}}ds.
\end{align}

Let $z:=\frac{1}{\sqrt{2}}A^{1/2}\left(  s+\frac{B}{A}\right)  $ so $dz=$
$\frac{1}{\sqrt{2}}A^{1/2}ds$ and $s=\frac{\sqrt{2}}{A^{1/2}}z-\frac{B}{A}$ to
transform the integral and obtain%
\begin{align}
\int_{0}^{\infty}se^{-\frac{1}{2}A\left(  s+\frac{B}{A}\right)  ^{2}}ds &
=\frac{1}{A}e^{-\frac{B^{2}}{2A}}-\sqrt{\frac{\pi}{2}}\frac{B}{A^{3/2}%
}erfc\left(  \frac{B}{\sqrt{2}A^{1/2}}\right)  \\
erfc\left(  z\right)   &  :=1-erf\left(  z\right)  =\frac{2}{\sqrt{\pi}}%
\int_{z}^{\infty}e^{-u^{2}}du.\nonumber
\end{align}

We have then the identity
\begin{gather}
\frac{d}{dx}\mathbb{P}\left\{  X_{1},X_{2}>0\ and\ \frac{X_{1}}{X_{2}}\leq
x\right\}  =\frac{\left(  2\pi\right)  ^{-1}}{\sigma_{1}\sigma_{2}\sqrt
{1-\rho^{2}}}\int_{0}^{\infty}sdse^{-\frac{1}{2}Q\left(  s\right)
}\nonumber\\
=\frac{\left(  2\pi\right)  ^{-1}}{\sigma_{1}\sigma_{2}\sqrt{1-\rho^{2}}%
}e^{\frac{B^{2}}{2A}-\frac{C}{2}}\left\{  \frac{1}{A}e^{-\frac{B^{2}}{2A}%
}-\sqrt{\frac{\pi}{2}}\frac{B}{A^{3/2}}erfc\left(  \frac{B}{\sqrt{2}A^{1/2}%
}\right)  \right\}
\end{gather}

Hence, the density of $X_{1}/X_{2}$ conditioned on $X_{1},X_{2}>0$ is given by
the basic relation (see \cite{BI})
\begin{align}
&  f_{X_{1}/X_{2}}\left(  x\ |\ Q_{1}\right) \nonumber\\
&  =\frac{\left(  2\pi\right)  ^{-1}}{\mathbb{P}\left(  Q_{1}\right)
\sigma_{1}\sigma_{2}\sqrt{1-\rho^{2}}}e^{\frac{B^{2}}{2A}-\frac{C}{2}}\left\{
\frac{1}{A}e^{-\frac{B^{2}}{2A}}-\sqrt{\frac{\pi}{2}}\frac{B}{A^{3/2}%
}erfc\left(  \frac{B}{\sqrt{2}A^{1/2}}\right)  \right\}
\end{align}
where $Q_{1}$ is the set $X_{1}>0,$ $X_{2}>0.$ The calculation of
$\mathbb{P}\left(  Q_{1}\right)  $ is carried out in Appendix A.

We can use $\omega:=B/\left(  2A\right)  ^{1/2}$ to write the following exact
expression for the conditional density:
\begin{align}
f_{X_{1}/X_{2}}\left(  x\ |\ Q_{1}\right)   &  =\frac{\left(  2\pi\right)
^{-1}}{\mathbb{P}\left(  Q_{1}\right)  \sigma_{1}\sigma_{2}\sqrt{1-\rho^{2}}%
}\frac{1}{A}e^{-\frac{C}{2}}h\left(  \omega\right)  ,\ \ \nonumber\\
h\left(  \omega\right)   &  :=e^{\omega^{2}}\left\{  e^{-\omega^{2}}-\sqrt
{\pi}\omega~erfc\left(  \omega\right)  \right\}  =1-\sqrt{\pi}\omega
e^{\omega^{2}}~erfc\left(  \omega\right)  .
\end{align}

\bigskip

\section{\ Analysis of the logarithm of the conditional density}

In order to extract the behavior of the conditional density for large
$\left\vert x\right\vert $ we analyze its logarithm.

\begin{theorem}
For $x\geq x_{0}:=2\frac{\sigma_{1}}{\sigma_{2}}$ one has the bounds%
\begin{align}
\frac{\log f_{X_{1}/X_{2}}\left(  x\ |\ Q_{1}\right)  }{\log x}  &
=-\frac{\frac{C}{2}}{\log x}+\frac{\log\left[  \frac{\left(  2\pi\right)
^{-1}}{\mathbb{P}\left(  Q_{1}\right)  }\right]  }{\log x}+\frac{\log h\left(
\omega\right)  }{\log x}-\frac{\log\left[  \sigma_{1}\sigma_{2}\sqrt
{1-\rho^{2}}A\right]  }{\log x}\nonumber\\
&  =-\frac{\frac{C}{2}}{\log x}+\frac{\log\left[  \frac{\left(  2\pi\right)
^{-1}}{\mathbb{P}\left(  Q_{1}\right)  }\right]  }{\log x}-\frac{\log\left(
\frac{\sigma_{2}}{\sigma_{1}}\frac{1}{\sqrt{1-\rho^{2}}}\right)  }{\log
x}\nonumber\\
&  +\frac{\log h\left(  \omega\right)  }{\log x}-\left(  2+2\frac
{-\frac{\sigma_{1}}{\sigma_{2}}\rho\frac{1}{x}}{\log x}\right)  +R\left(
x;\frac{\sigma_{1}}{\sigma_{2}}\right)  ,
\end{align}
where $\left\vert R\left(  x;\frac{\sigma_{1}}{\sigma_{2}}\right)  \right\vert
\leq\frac{\left(  \frac{\sigma_{1}}{\sigma_{2}}\right)  ^{2}}{x_{0}^{2}\log
x_{0}}$\ .
\end{theorem}

\begin{proof}
Re-grouping the constants, and taking the logarithm of $f_{X_{1}/X_{2}}\left(
x\ |\ Q_{1}\right)  $, we write%
\begin{equation}
\log f_{X_{1}/X_{2}}\left(  x\ |\ Q_{1}\right)  =-\frac{C}{2}+\log\left[
\frac{\left(  2\pi\right)  ^{-1}}{\mathbb{P}\left(  Q_{1}\right)  }\right]
+\log h\left(  \omega\right)  -\log\left[  \sigma_{1}\sigma_{2}\sqrt
{1-\rho^{2}}A\right]  . \label{logFromThm4.1}%
\end{equation}
The decay exponent will be determined by the large $x$ limit of this quantity
divided by $\log x.$ We first analyze the last term in the expression above:%
\begin{equation}
\sigma_{1}\sigma_{2}\sqrt{1-\rho^{2}}A=\frac{\sigma_{2}}{\sigma_{1}}%
\frac{\left(  x-\frac{\sigma_{1}}{\sigma_{2}}\rho\right)  ^{2}}{\sqrt
{1-\rho^{2}}}\left[  1+\frac{\left(  \frac{\sigma_{1}}{\sigma_{2}}\right)
^{2}\left(  1-\rho^{2}\right)  }{\left(  x-\frac{\sigma_{1}}{\sigma_{2}}%
\rho\right)  ^{2}}\right]
\end{equation}
Taking the logarithm and dividing by $\log x,$ we have%
\begin{align}
&  \frac{\log\left[  \sigma_{1}\sigma_{2}\sqrt{1-\rho^{2}}A\right]  }{\log
x}\nonumber\\
&  =\frac{\log\left[  \frac{\sigma_{2}}{\sigma_{1}}\frac{1}{\sqrt{1-\rho^{2}}%
}\right]  }{\log x}+2\frac{\log\left(  x-\frac{\sigma_{1}}{\sigma_{2}}%
\rho\right)  }{\log x}+\frac{\log\left[  1+\frac{\left(  \frac{\sigma_{1}%
}{\sigma_{2}}\right)  ^{2}\left(  1-\rho^{2}\right)  }{\left(  x-\frac
{\sigma_{1}}{\sigma_{2}}\rho\right)  ^{2}}\right]  }{\log x} \label{logThm4.1}%
\end{align}
We examine the last two terms for $x\geq2x_{0}$%
\begin{align}
2\frac{\log\left(  x-\frac{\sigma_{1}}{\sigma_{2}}\rho\right)  }{\log x}  &
=2+2\frac{\log\left(  1-\frac{\sigma_{1}}{\sigma_{2}}\rho\frac{1}{x}\right)
}{\log x}\nonumber\\
\left\vert 2\frac{\log\left(  x-\frac{\sigma_{1}}{\sigma_{2}}\rho\right)
}{\log x}-\left(  2+2\frac{-\frac{\sigma_{1}}{\sigma_{2}}\rho\frac{1}{x}}{\log
x}\right)  \right\vert  &  \leq\frac{1}{2}\frac{\left(  \frac{\sigma_{1}%
}{\sigma_{2}}\rho\right)  ^{2}}{x_{0}^{2}\log x_{0}}\leq\frac{1}{2}R\left(
x_{0};\frac{\sigma_{1}}{\sigma_{2}}\right)
\end{align}
The last logarithm in (\ref{logThm4.1}) can be bounded (for $x\geq2x_{0}$) as
\begin{equation}
\left\vert \frac{\log\left[  1+\frac{\left(  \frac{\sigma_{1}}{\sigma_{2}%
}\right)  ^{2}\left(  1-\rho^{2}\right)  }{\left\{  x-\frac{\sigma_{1}}%
{\sigma_{2}}\rho\right\}  ^{2}}\right]  }{\log x}\right\vert \leq\frac{\left(
\frac{\sigma_{1}}{\sigma_{2}}\right)  ^{2}\left(  1-\rho^{2}\right)
}{\left\{  x_{0}-\frac{\sigma_{1}}{\sigma_{2}}\rho\right\}  ^{2}\log x_{0}%
}\leq\frac{1}{2}R\left(  x_{0};\frac{\sigma_{1}}{\sigma_{2}}\right)
\end{equation}
Thus one has%
\begin{align}
&  \left\vert -\frac{\log\left[  \sigma_{1}\sigma_{2}\sqrt{1-\rho^{2}%
}A\right]  }{\log x}-\left[  -2+2\frac{\frac{\sigma_{1}}{\sigma_{2}}\rho}{\log
x}+\frac{\log\left[  \frac{\sigma_{2}}{\sigma_{1}}\frac{1}{\sqrt{1-\rho^{2}}%
}\right]  }{\log x}\right]  \right\vert \nonumber\\
&  \leq R\left(  x_{0};\frac{\sigma_{1}}{\sigma_{2}}\right)  ,
\end{align}
concluding the proof.
\end{proof}

In order to extract the relevant part of $\log\omega$ we write
\begin{equation}
\omega_{0}:=\frac{1}{2^{1/2}}\left(  1-\rho^{2}\right)  ^{-1/2}\left(
\rho\frac{\mu_{2}}{\sigma_{2}}-\frac{\mu_{1}}{\sigma_{1}}\right)  ,\text{
}-\frac{C}{2}=-\frac{1}{2}\frac{\mu_{2}^{2}}{\sigma_{2}^{2}}-\omega_{0}^{2},
\end{equation}
and use this to extract the $x-$dependent part of the conditional density.

\bigskip

\begin{theorem}
(a) For $x>x_{0}:=\max\left\{  2\frac{\sigma_{1}}{\sigma_{2}},1\right\}  $ one
has the bounds%
\begin{align}
&  \frac{\log f_{X_{1}/X_{2}}\left(  x\ |\ Q_{1}\right)  }{\log x}\nonumber\\
&  =-\frac{\frac{C}{2}}{\log x}+\frac{\log\left[  \frac{\left(  2\pi\right)
^{-1}}{\mathbb{P}\left(  Q_{1}\right)  }\right]  }{\log x}+\frac{\log h\left(
\omega\right)  }{\log x}-\frac{\log\left[  \sigma_{1}\sigma_{2}\sqrt
{1-\rho^{2}}A\right]  }{\log x}\nonumber\\
&  =-2-\frac{\frac{C}{2}}{\log x}+\frac{\log\left[  \frac{\left(  2\pi\right)
^{-1}}{\mathbb{P}\left(  Q_{1}\right)  }\right]  }{\log x}-\frac{\log\left(
\frac{\sigma_{2}}{\sigma_{1}}\frac{1}{\sqrt{1-\rho^{2}}}\right)  }{\log
x}+\frac{\log h\left(  \omega_{0}\right)  }{\log x}+R_{4}\left(  x_{0}\right)
\end{align}
where $R_{4}$ satisfies $\left\vert R_{4}\left(  x_{0}\right)  \right\vert
\leq$\ $\frac{C_{1}}{x_{0}\log x_{0}}$ and more specifically
\begin{equation}
\left\vert R_{4}\left(  x_{0}\right)  \right\vert \leq R_{1}\left(
x_{0};\frac{\sigma_{1}}{\sigma_{2}}\left\vert \rho\right\vert \right)
+R_{2}\left(  x_{0};\frac{\sigma_{1}}{\sigma_{2}}\right)  +R_{3}\left(
x_{0};\left\vert \frac{b_{1}}{a_{1}}\right\vert ,a_{2}\right)
\end{equation}
where$\ \ R_{1}\left(  x;\frac{\sigma_{1}}{\sigma_{2}}\left\vert
\rho\right\vert \right)  :=2\frac{\frac{\sigma_{1}}{\sigma_{2}}\left\vert
\rho\right\vert }{x\log x}$ \ \ $R_{2}\left(  x;\frac{\sigma_{1}}{\sigma_{2}%
}\right)  :=2\frac{\left(  \frac{\sigma_{1}}{\sigma_{2}}\right)  ^{2}}{x^{2}}$
and
\begin{equation}
R_{3}\left(  x;\left\vert \frac{b_{1}}{a_{1}}\right\vert ,a_{2}\right)
:=\left\vert \omega_{0}\right\vert \left\{  \left\vert \frac{b_{1}}{a_{1}%
}\right\vert +\frac{1}{2a_{2}^{2}}+\frac{1}{8a_{2}^{4}}\right\}  \left\vert
h^{\prime}\left(  \zeta\right)  \right\vert \frac{1}{x_{0}\log x_{0}}\leq
\frac{C_{1}}{x_{0}\log x_{0}}%
\end{equation}
where $a_{1},a_{2},b_{1}$ and $\zeta$ are defined below.

(b) In the limit one has%
\begin{equation}
\lim_{x\rightarrow\infty}\frac{\log f_{X_{1}/X_{2}}\left(  x\ |\ Q_{1}\right)
}{\log x}=-2\ ,
\end{equation}
and the density can be expressed in the form%
\begin{equation}
f_{X_{1}/X_{2}}\left(  x\ |\ Q_{1}\right)  \simeq f_{0}x^{-2}%
\end{equation}%
\[
f_{0}:=\frac{\left(  \sigma_{1}/\sigma_{2}\right)  \sqrt{1-\rho^{2}}}%
{2\pi\mathbb{P}\left(  Q_{1}\right)  }e^{-\frac{1}{2}\mu_{2}^{2}/\sigma
_{2}^{2}}\left(  e^{-\omega_{0}^{2}}-\sqrt{\pi}\omega_{0}erfc\left(
\omega_{0}\right)  \right)  .
\]

\end{theorem}

\begin{proof}
The main issue is to examine the two terms, $\log\left[  \sigma_{1}\sigma
_{2}\sqrt{1-\rho^{2}}A\right]  $ and $\log h\left(  \omega\right)  $ and
extract the part that is significant after division by $\log x$.

$\left(  i\right)  $ One has
\begin{align}
\log\left[  \sigma_{1}\sigma_{2}\sqrt{1-\rho^{2}}A\right]   &  =\log\left(
\frac{\sigma_{2}}{\sigma_{1}}\frac{1}{\sqrt{1-\rho^{2}}}\right)  +2\log\left(
x-\frac{\sigma_{1}}{\sigma_{2}}\rho\right) \nonumber\\
&  +\log\left(  1+\frac{\left(  \frac{\sigma_{1}}{\sigma_{2}}\right)
^{2}\left(  1-\rho^{2}\right)  }{\left(  x-\frac{\sigma_{1}}{\sigma_{2}}%
\rho\right)  }\right)  . \label{logThm4.2}%
\end{align}
Upon dividing by $\log x$ the second of these terms is%
\begin{equation}
\frac{2\log\left(  x-\frac{\sigma_{1}}{\sigma_{2}}\rho\right)  }{\log
x}=2+2\frac{\log\left(  1-\frac{\sigma_{1}}{\sigma_{2}}\rho\frac{1}{x}\right)
}{\log x},
\end{equation}
and the latter term is bounded by%
\begin{equation}
2\left\vert \frac{\log\left(  1-\frac{\sigma_{1}}{\sigma_{2}}\rho\frac{1}%
{x}\right)  }{\log x}\right\vert \leq2\frac{\frac{\sigma_{1}}{\sigma_{2}%
}\left\vert \rho\right\vert }{x\log x}\leq R_{1}\left(  x_{0};\frac{\sigma
_{1}}{\sigma_{2}}\right)  .
\end{equation}
The last of the terms in (\ref{logThm4.2}) can be bounded as%
\begin{align}
0  &  \leq\log\left(  1+\frac{\left(  \frac{\sigma_{1}}{\sigma_{2}}\right)
^{2}\left(  1-\rho^{2}\right)  }{\left(  x-\frac{\sigma_{1}}{\sigma_{2}}%
\rho\right)  }\right)  \leq\frac{1}{2}\frac{\left(  \frac{\sigma_{1}}%
{\sigma_{2}}\right)  ^{2}\left(  1-\rho^{2}\right)  }{\left(  x-\frac
{\sigma_{1}}{\sigma_{2}}\rho\right)  ^{2}}\nonumber\\
&  \leq\frac{2\left(  \frac{\sigma_{1}}{\sigma_{2}}\right)  ^{2}}{x^{2}}.
\end{align}
Hence, one has the bounds%
\begin{equation}
\frac{\log\left(  1+\frac{\left(  \frac{\sigma_{1}}{\sigma_{2}}\right)
^{2}\left(  1-\rho^{2}\right)  }{\left(  x-\frac{\sigma_{1}}{\sigma_{2}}%
\rho\right)  }\right)  }{\log x}\leq R_{2}\left(  x_{0};\frac{\sigma_{1}%
}{\sigma_{2}}\right)  .
\end{equation}%
\begin{equation}
\left\vert -\frac{\log\left[  \sigma_{1}\sigma_{2}\sqrt{1-\rho^{2}}A\right]
}{\log x}+2\right\vert \leq R_{1}\left(  x_{0};\frac{\sigma_{1}}{\sigma_{2}%
}\right)  +R_{2}\left(  x_{0};\frac{\sigma_{1}}{\sigma_{2}}\right)  \ .
\end{equation}
$\left(  ii\right)  $ Next, we examine $\omega-\omega_{0}$ as we focus on
$\log h\left(  \omega\right)  .$ The term $h\left(  \omega\right)  $ has $x$
dependence. Note that $\omega$ can be written as%
\begin{align}
\omega &  \dot{=}\frac{B}{2^{1/2}A^{1/2}}=\frac{-\sigma_{2}^{-2}\mu_{2}%
+\sigma_{1}^{-2}\left(  1-\rho^{2}\right)  ^{-1}\left(  x-\rho\frac{\sigma
_{1}}{\sigma_{2}}\right)  \left(  \mu_{2}\rho\frac{\sigma_{1}}{\sigma_{2}}%
-\mu_{1}\right)  }{\left\{  \sigma_{2}^{-2}+\sigma_{1}^{-2}\left(  1-\rho
^{2}\right)  ^{-1}\left(  x-\rho\frac{\sigma_{1}}{\sigma_{2}}\right)
^{2}\right\}  ^{1/2}}\nonumber\\
&  =\frac{1}{2^{1/2}}\frac{a_{1}+\frac{b_{1}}{x-\rho\sigma_{1}/\sigma_{2}}%
}{\left(  a_{2}^{2}+\frac{1}{\left(  x-\rho\sigma_{1}/\sigma_{2}\right)  ^{2}%
}\right)  ^{1/2}}%
\end{align}
and $\omega_{0}=2^{-1/2}a_{1}/a_{2}$ with the definitions
\begin{align}
a_{1}  &  :=\left(  \frac{\mu_{2}}{\sigma_{2}}\rho-\frac{\mu_{1}}{\sigma_{1}%
}\right)  \frac{\sigma_{2}}{\sigma_{1}}\left(  1-\rho^{2}\right)
^{-1},\ \ b_{1}:=-\frac{\mu_{2}}{\sigma_{2}}\nonumber\\
a_{2}^{2}  &  :=\left(  \frac{\sigma_{2}}{\sigma_{1}}\right)  \left(
1-\rho^{2}\right)  ^{-1}.
\end{align}
Using $P:=\frac{b_{1}}{a_{1}}\frac{1}{x-\rho\sigma_{1}/\sigma_{2}}$ and
$Q^{2}:=\left[  a_{2}^{2}(x-\rho\sigma_{1}/\sigma_{2})^{2}\right]  ^{-1}$ we
write%
\[
\left\vert \omega-\omega_{0}\right\vert =\left\vert \omega_{0}\right\vert
\left\vert \frac{1+P-\left(  1+Q^{2}\right)  ^{1/2}}{\left(  1+Q^{2}\right)
^{1/2}}\right\vert
\]
Basic Taylor series estimates then imply%
\begin{equation}
\left\vert 1+P-\left(  1+Q^{2}\right)  ^{1/2}\right\vert \leq\left\vert
P\right\vert +\frac{1}{2}Q^{2}+\frac{1}{8}Q^{4}.
\end{equation}
Upon using $x\geq2x_{0}\ $so that $x-\rho\sigma_{1}/\sigma_{2}\geq x_{0}$ , we
have the bounds%
\begin{equation}
\left\vert \omega-\omega_{0}\right\vert \leq\left\vert \omega_{0}\right\vert
\left\{  \left\vert \frac{b_{1}}{a_{1}}\right\vert \frac{1}{x_{0}}+\frac
{1}{2a_{2}^{2}}\frac{1}{x_{0}^{2}}+\frac{1}{8}\frac{1}{a_{2}^{4}x_{0}^{4}%
}\right\}  .
\end{equation}
Thus, for large $x$ we can thus write, for $x\geq2x_{0}$
\begin{equation}
\left\vert \omega-\omega_{0}\right\vert \leq\frac{C_{1}}{x_{0}}%
\end{equation}
where $C_{1}$ depends on $\left(  1-\rho^{2}\right)  ,$ $\sigma_{1}/\sigma
_{2}$ and $\left\vert \rho\frac{\mu_{2}}{\sigma_{2}}-\frac{\mu_{1}}{\sigma
_{1}}\right\vert .$

This yields the bound
\begin{equation}
\left\vert \log h\left(  \omega\right)  -\log h\left(  \omega_{0}\right)
\right\vert \leq\left\vert \omega-\omega_{0}\right\vert \ \left\vert
h^{\prime}\left(  \zeta\right)  \right\vert
\end{equation}
where $\zeta:=\inf\left\{  \omega,\omega_{0}\right\}  $. Hence, we have%
\begin{equation}
\left\vert \frac{\log h\left(  \omega\right)  }{\log x}-\frac{\log h\left(
\omega_{0}\right)  }{\log x}\right\vert \leq R_{3}\left(  x_{0};\omega
_{0},\left\vert b_{1}/a_{1}\right\vert ,a_{2}\right)
\end{equation}
where%
\begin{equation}
R_{3}\left(  x_{0};\omega_{0},\left\vert b_{1}/a_{1}\right\vert ,a_{2}\right)
\leq\left\{  \left\vert \frac{b_{1}}{a_{1}}\right\vert +\frac{1}{2a_{2}^{2}%
}+\frac{1}{8a_{2}^{4}}\right\}  \frac{1}{x_{0}}.
\end{equation}
That is, one has the bound, with $C_{2}$ depending on $\omega_{0},\left\vert
b_{1}/a_{1}\right\vert ,a_{2}$ and $\left\vert h^{\prime}\left(  \zeta\right)
\right\vert $,
\begin{equation}
\left\vert \frac{\log h\left(  \omega\right)  }{\log x}-\frac{\log h\left(
\omega_{0}\right)  }{\log x}\right\vert \leq\ \frac{C_{2}}{x_{0}\log x_{0}%
}\left\vert h^{\prime}\left(  \zeta\right)  \right\vert .
\end{equation}
\ The proof is concluded by observing that the $h^{\prime}$ term can be
bounded using the following classical estimates for the error function:
\begin{equation}
\frac{1}{\omega+\sqrt{\omega^{2}+2}}<e^{\omega^{2}}\int_{\omega}^{\infty
}e^{-u^{2}}du\leq\frac{1}{\omega+\sqrt{\omega^{2}+\frac{4}{\pi}}}%
\end{equation}
By basic error function series expansions, one can bound $h\left(
\omega\right)  $ by%
\[
1-\frac{2}{1+\sqrt{1+\frac{2}{\omega^{2}}}}\leq h\left(  \omega\right)
\leq1-\frac{2}{1+\sqrt{1+\frac{4}{\pi\omega^{2}}}}%
\]
For $\omega\geq0,$ one has the additional bounds (see \cite{AB})%
\begin{equation}
2\omega-\frac{2\left(  1+2\omega^{2}\right)  }{\omega+\sqrt{\omega^{2}+4/\pi}%
}\leq h^{\prime}\left(  \omega\right)  \leq2\omega-\frac{2\left(
1+2\omega^{2}\right)  }{\omega+\sqrt{\omega^{2}+2}}\ .
\end{equation}

\end{proof}

We will consider both of the regions $-\omega_{0}>>1$ and $\omega_{0}>>1$
separately after considering the special case of $\rho:=-1.$

\subsection{The Singular Case $\rho:=-1.$}

Recall that our calculations have been under the assumption of a nonsingular
covariance matrix $\Sigma$, i.e. where $\left\vert \rho\right\vert <1$.
Another interesting case is when $\rho=-1$, which corresponds to two
anticorrelated random variables. This is useful in modeling asset pricing as
an increase in demand will often be accompanied by a commensurate decrease in
supply and vice versa. To this end, let $Y\sim\mathcal{N}\left(  0,1\right)  $
be a standard normal random variable and suppose demand and supply then take
the form, for positive $\mu_{1},\mu_{2},\sigma_{1},\sigma_{2},$%
\begin{align}
D  &  =X_{1}\sim N\left(  \mu_{1},\sigma_{1}^{2}\right)  =\mu_{1}+\sigma
_{1}Y\nonumber\\
S  &  =X_{2}\sim N\left(  \mu,\sigma^{2}\right)  =\mu_{1}-\sigma_{2}Y.
\end{align}
Thus, we will have the covariance matrix%
\begin{equation}
\sum=\left(
\begin{array}
[c]{cc}%
\sigma_{1}^{2} & -\sigma_{1}\sigma_{2}\\
-\sigma_{1}\sigma_{2} & \sigma_{2}^{2}%
\end{array}
\right)  .
\end{equation}
Then the price equation becomes%
\[
\frac{1}{P}\frac{dP}{dt}=\frac{D}{S}-1=\frac{\mu_{1}+\sigma_{1}Y}{\mu
_{1}-\sigma_{2}Y}-1=:R-1
\]

Under near-equilibrium conditions, the values $\mu_{1}/\mu_{2}$ and
$\sigma_{1}/\sigma_{2}$ will be near unity. Nevertheless, we prove the
following more general theorem.

\begin{theorem}
Let $\mu_{1},\mu_{2},\sigma_{1},\sigma_{2}$ be strictly positive, $X_{1}%
,X_{2}$ be bivariate normals with correlation $-1.$ Then $R:=X_{1}/X_{2}$ has
density%
\begin{equation}
f_{R}\left(  x\right)  =\frac{\mu_{1}\sigma_{2}+\mu_{2}\sigma_{1}}{\sqrt{2\pi
}}\frac{e^{-\frac{1}{2}\left(  \frac{\mu_{2}x-\mu_{1}}{\sigma_{2}x+\sigma_{1}%
}\right)  ^{2}}}{\left(  \sigma_{2}x+\sigma_{1}\right)  ^{2}}%
\end{equation}
where $f_{R}\left(  -\sigma_{1}/\sigma_{2}\right)  :=0.$

\begin{proof}
Let $Y\sim N\left(  0,1\right)  $ be the standard normal. With $R=\left(
\mu_{1}+\sigma_{1}Y\right)  /\left(  \mu_{2}-\sigma_{2}Y\right)  $ we use the
Theorem of Total Probability to express the distribution function of $R$ as%
\begin{align*}
\mathbb{P}\left\{  R\leq x\right\}   &  =\mathbb{P}\left\{  \frac{\mu
_{1}+\sigma_{1}Y}{\mu_{2}-\sigma_{2}Y}\leq x,\mu_{2}-\sigma_{2}Y>0\right\} \\
&  +\mathbb{P}\left\{  \frac{\mu_{1}+\sigma_{1}Y}{\mu_{2}-\sigma_{2}Y}\leq
x,\mu_{2}-\sigma_{2}Y<0\right\}
\end{align*}
where we neglect sets of measure zero. Abbreviating the two probabilities by
$P_{1}$ and $P_{2}$ respectively, we note the two cases determined by the sign
of $\sigma_{2}x+\sigma_{1}.$ $\left(  i\right)  $ For $\sigma_{2}x+\sigma
_{1}>0$ we have
\[
P_{1}=\mathbb{P}\left\{  Y\leq\frac{\mu_{2}x-\mu_{1}}{\sigma_{2}x+\sigma_{1}%
}\ and\ Y<\frac{\mu_{2}}{\sigma_{2}}\right\}  =\int_{-\infty}^{\frac{\mu
_{2}x-\mu_{1}}{\sigma_{2}x+\sigma_{1}}}f_{Y}\left(  s\right)  ds,
\]%
\[
P_{2}=\mathbb{P}\left\{  Y\geq\frac{\mu_{2}x-\mu_{1}}{\sigma_{2}x+\sigma_{1}%
}\ and\ Y>\frac{\mu_{2}}{\sigma_{2}}\right\}  =\int_{\frac{\mu_{2}}{\sigma
_{2}}}^{\infty}f_{Y}\left(  s\right)  ds.
\]
Differentiating the sum with respect to $x,$ we note that the latter vanishes,
and we obtain, for $\sigma_{2}x+\sigma_{1}>0,$ the density%
\[
f_{R}\left(  x\right)  =f_{Y}\left(  \frac{\mu_{2}x-\mu_{1}}{\sigma
_{2}x+\sigma_{1}}\right)  \frac{d}{dx}\left(  \frac{\mu_{2}x-\mu_{1}}%
{\sigma_{2}x+\sigma_{1}}\right)  =\frac{e^{\frac{-1}{2}\left(  \frac{\mu
_{2}x-\mu_{1}}{\sigma_{2}x+\sigma_{1}}\right)  ^{2}}}{\sqrt{2\pi}}\frac
{\mu_{1}\sigma_{2}+\mu_{2}\sigma_{1}}{\left(  \sigma_{2}x+\sigma_{1}\right)
^{2}}.
\]
$\left(  ii\right)  \ $Next, for $\sigma_{2}x+\sigma_{1}<0$ we use the same
notation to write%
\begin{align*}
P_{1}  &  =\mathbb{P}\left\{  Y\geq\frac{\mu_{2}x-\mu_{1}}{\sigma_{2}%
x+\sigma_{1}}\ and\ Y<\frac{\mu_{2}}{\sigma_{2}}\right\}  =\int_{\frac{\mu
_{2}x-\mu_{1}}{\sigma_{2}x+\sigma_{1}}}^{\frac{\mu_{2}}{\sigma_{2}}}%
f_{Y}\left(  s\right)  ds,\\
P_{2}  &  =\mathbb{P}\left\{  Y\leq\frac{\mu_{2}x-\mu_{1}}{\sigma_{2}%
x+\sigma_{1}}\ and\ Y>\frac{\mu_{2}}{\sigma_{2}}\right\}  =0.
\end{align*}
so that differentiation with respect to $x$ yields the same result as above,
proving the theorem.
\end{proof}

\begin{remark}
In the special case $\mu_{1}=\mu_{2}=\mu$ and $\ \sigma_{1}=\sigma_{2}=\sigma$
the density has the form
\begin{equation}
f_{R}\left(  x\right)  =\sqrt{\frac{2}{\pi}}\frac{\mu}{\sigma}\frac
{e^{-\frac{1}{2}\left(  \frac{\mu}{\sigma}\frac{x-1}{x+1}\right)  ^{2}}%
}{\left(  x+1\right)  ^{2}} \label{exactDensity}%
\end{equation}
where $f_{R}\left(  -1\right)  :=0.$
\end{remark}
\end{theorem}

\subsection{The Cauchy Limit}

We consider the limit in which $\rho:=0$ and $\mu_{1},\mu_{2}\rightarrow0+$
with $\sigma_{1},\sigma_{2}$ fixed at $1$, $\omega=0,$ $-C/2=-\frac{1}{2}%
\frac{\mu_{2}^{2}}{\sigma_{2}^{2}}-\omega_{0}^{2}=0,$ $\log h\left(
\omega_{0}\right)  =0.$ Since $\rho:=0,$ the random variables $X_{1}$ and
$X_{2}$ are independent, and together with $\mu_{1}=\mu_{2}=0,$ one has
\begin{equation}
\mathbb{P}\left(  Q_{1}\right)  =1/4,~\ \sigma_{1}\sigma_{2}\sqrt{1-\rho^{2}%
}A=1+x^{2}.
\end{equation}
Thus, the expression (\ref{logFromThm4.1}) simplifies to
\begin{equation}
\log f_{X_{1}/X_{2}}\left(  x\ |\ Q_{1}\right)  =\log\frac{2}{\pi\left(
1+x^{2}\right)  }\ .
\end{equation}

Hence, this is a simple proof of the known result that the quotient of two
independent standard normals (i.e., mean $0$ and variance $1$) yields the
Cauchy density%
\begin{equation}
f_{Cauchy}\left(  x\ |\ x>0\right)  =\frac{2}{\pi}\frac{1}{1+x^{2}}.
\end{equation}
Thus, the limit $\mu_{1},\mu_{2}\rightarrow0+$ with $\sigma_{1},\sigma_{2}$
set at $1$ and $\rho$ at $0,$ recovers the classical limit of the Cauchy
density, and one obtains
\begin{equation}
\lim_{x\rightarrow\infty}\frac{\log f_{X_{1}/X_{2}}\left(  x\ |\ Q_{1}\right)
}{\log x}=\lim_{x\rightarrow\infty}\frac{\log f_{Cauchy}\left(  x\ |\ Q_{1}%
\right)  }{\log x}=-2.
\end{equation}

\subsection{The limit of constant denominator}

We consider for positive $X_{1}$ and $X_{2}$ the ratio $X_{1}/X_{2}$ in the
limit in which the denominator approaches a constant, i.e., $\mu_{2}>0$ and
$\sigma_{2}\rightarrow0.$ Note that when $\mu_{1}$ and $\mu_{2}$ are positive
and $\mu_{1}/\sigma_{1}$ and $\mu_{2}/\sigma_{2}$ are large, there is a very
small difference between the density and conditional density.

The basic tool we use is an asymptotic expression for the integral for $b>0$
and $a>>1,$%
\begin{align}
I\left(  a,b\right)   &  :=\int_{0}^{\infty}e^{-a\left(  s-b\right)  ^{2}%
}g\left(  s\right)  ds\tilde{=}\int_{-\infty}^{\infty}e^{-a\left(  s-b\right)
^{2}}g\left(  s\right)  ds\nonumber\\
&  \tilde{=}\int_{-\infty}^{\infty}e^{-a\left(  s-b\right)  ^{2}}g\left(
b\right)  ds.
\end{align}
In other words, the integral will be negligible when $s$ is far from $b,$
since $a>>1.$ Let $z:=a^{1/2}\left(  s-b\right)  $ and $dz=a^{1/2}ds,$ so we
can write%
\begin{equation}
I\left(  a,b\right)  \tilde{=}g\left(  b\right)  \int_{-\infty}^{\infty
}e^{-z^{2}}a^{-1/2}dz=\left(  \frac{\pi}{a}\right)  ^{1/2}g\left(  b\right)  .
\end{equation}
More precisely, for a function $g$ such that $\left\vert g\right\vert \leq1$
and $\left\vert g^{\prime\prime}\right\vert \leq M,$ one can use Taylor series
bounds to obtain%
\begin{equation}
\left\vert I\left(  a,b\right)  -\left(  \frac{\pi}{a}\right)  ^{1/2}g\left(
b\right)  \right\vert \leq\frac{\sqrt{\pi}}{4}Ma^{-3/2}+\frac{1}{a^{1/2}%
b}e^{-ab}\ .
\end{equation}
We start from the early expression (\ref{Q1half}), namely,%
\begin{align}
&  f_{X_{1}/X_{2}}\left(  x~|\ Q_{1}\right)  :=\frac{d}{dx}\mathbb{P}\left\{
X_{1},X_{2}>0\ and\ \frac{X_{1}}{X_{2}}\leq x\right\} \nonumber\\
&  =\frac{\left(  2\pi\right)  ^{-1}}{\sigma_{1}\sigma_{2}\sqrt{1-\rho^{2}}%
}\int_{0}^{\infty}s_{2}ds_{2}e^{-\frac{1}{2}\left(  \frac{s_{2}-\mu_{2}%
}{\sigma_{2}}\right)  ^{2}}e^{-\frac{1}{2}\left(  \frac{xs_{2}-\left(  \mu
_{1}+\frac{\rho\sigma_{1}}{\sigma_{2}}\left(  s_{2}-\mu_{2}\right)  \right)
}{\sigma_{1}\sqrt{1-\rho^{2}}}\right)  ^{2}},
\end{align}
and apply the asymptotic result above with $g$ as the integrand above with
$a:=\left(  2\sigma_{2}^{2}\right)  ^{-1}$ and $b:=\mu_{2}$. This yields (with
$\rho:=0$), with $E_{1}\left(  a,b\right)  $ the error term,%
\begin{align}
&  \int_{0}^{\infty}s_{2}ds_{2}e^{-\frac{1}{2}\left(  \frac{s_{2}-\mu_{2}%
}{\sigma_{2}}\right)  ^{2}}e^{-\frac{1}{2}\left(  \frac{xs_{2}-\left(  \mu
_{1}+\frac{\rho\sigma_{1}}{\sigma_{2}}\left(  s_{2}-\mu_{2}\right)  \right)
}{\sigma_{1}\sqrt{1-\rho^{2}}}\right)  ^{2}}\nonumber\\
&  =\mu_{2}e^{-\frac{1}{2}\left(  \frac{x\mu_{2}-\mu_{1}}{\sigma_{1}}\right)
^{2}}\left(  2\pi\sigma_{2}^{2}\right)  ^{1/2}+E_{1}\left(  a,b\right)
\end{align}
so that substitution into the integral yields%
\begin{align}
f_{X_{1}/X_{2}}\left(  x~|\ Q_{1}\right)   &  =\frac{\left(  2\pi\right)
^{-1}}{\sigma_{1}\sigma_{2}}\mu_{2}e^{-\frac{1}{2}\left(  \frac{x\mu_{2}%
-\mu_{1}}{\sigma_{1}}\right)  ^{2}}\left(  2\pi\sigma_{2}^{2}\right)
^{1/2}+E_{1}\left(  \left(  2\sigma_{2}^{2}\right)  ^{-1},\mu_{2}\right)
\nonumber\\
&  =\frac{\mu_{2}}{\left(  2\pi\right)  ^{1/2}\sigma_{1}}e^{-\frac{1}%
{2}\left(  \frac{x\mu_{2}-\mu_{1}}{\sigma_{1}}\right)  ^{2}}+E_{1}\left(
\left(  2\sigma_{2}^{2}\right)  ^{-1},\mu_{2}\right)  .
\end{align}
The error term is bounded by%
\begin{align}
\left\vert E_{1}\left(  \left(  2\sigma_{2}^{2}\right)  ^{-1},\mu_{2}\right)
\right\vert  &  \leq\frac{1}{2\pi\sigma_{1}\sigma_{2}}\left(  \frac{\sqrt{\pi
}2^{3/2}}{4}M\sigma_{2}^{3/2}+\pi^{-1/2}\frac{2\sigma_{2}^{2}}{\mu_{2}%
}e^{-\left(  2\sigma_{2}^{2}\right)  ^{3/2}\mu_{2}}\right) \nonumber\\
&  =\frac{\sigma_{2}^{1/2}}{2\pi\sigma_{1}}\frac{\sqrt{\pi}2^{3/2}}{4}%
M+\frac{\sigma_{2}}{2\pi\sigma_{1}}\pi^{-1/2}\frac{2}{\mu_{2}}e^{-\left(
2\sigma_{2}^{2}\right)  ^{3/2}\mu_{2}},
\end{align}
and vanishes as $\sigma_{2}^{1/2}\rightarrow0.$

Recall that the density for a $\mathcal{N}\left(  \mu_{1},\sigma_{1}%
^{2}\right)  $ random variable is
\begin{equation}
f_{\mu_{1},\sigma_{1}^{2}}\left(  x\right)  =\frac{1}{\left(  2\pi\right)
^{1/2}\sigma_{1}}e^{-\frac{1}{2}\left(  x-\mu_{1}\right)  /\sigma_{1}^{2}}.
\end{equation}
Given random variables $X_{1}\sim\mathcal{N}\left(  \mu_{1},\sigma_{1}%
^{2}\right)  $ and $X_{2}\sim\mathcal{N}\left(  \mu_{2},\sigma_{2}^{2}\right)
,$ in the limit as $\sigma_{2}\rightarrow0,$ the denominator is simply
division by $\mu_{2}$ (since $X_{2}=\mu_{2}$ at that stage). Thus, we have
that%
\begin{equation}
\frac{1}{\mu_{2}}X_{1}\sim\mathcal{N}\left(  \frac{\mu_{1}}{\mu_{2}},\left(
\frac{1}{\mu_{2}}\right)  ^{2}\sigma_{1}^{2}\right)
\end{equation}
so the density of $X_{1}/\mu_{2}$ is given by%
\begin{equation}
f_{X_{1}/\mu_{2}}\left(  x\right)  =\mu_{2}\frac{1}{\left(  2\pi\right)
^{1/2}\sigma_{1}}e^{-\frac{1}{2}\left(  \mu_{2}x-\mu_{1}\right)  ^{2}%
/\sigma_{1}^{2}}.
\end{equation}
Hence, we see that $f_{X_{1}/\mu_{2}}\left(  x\right)  =f_{X_{1}/X_{2}}\left(
x~|\ Q_{1}\right)  $ in the limit as $\sigma_{2}\rightarrow0$ with $\rho:=0,$
as expected.

\subsection{The limits $-\omega_{0}>>1$ and $\omega_{0}>>1.$}

Recalling the definition of $\omega_{0}$ and $C,$ we have%
\begin{equation}
C=\frac{\mu_{2}^{2}}{\sigma_{2}^{2}}+\left(  1-\rho^{2}\right)  ^{-1}\left(
\rho\frac{\mu_{2}}{\sigma_{2}}-\frac{\mu_{1}}{\sigma_{1}}\right)
^{2},\ \ \omega_{0}:=\frac{1}{2^{1/2}}\left(  1-\rho^{2}\right)
^{-1/2}\left(  \rho\frac{\mu_{2}}{\sigma_{2}}-\frac{\mu_{1}}{\sigma_{1}%
}\right)
\end{equation}
Rewriting \ref{logFromThm4.1} we have the exact expression
\begin{equation}
\log f_{X_{1}/X_{2}}\left(  x\ |\ Q_{1}\right)  =-\frac{C}{2}+\log\left[
\frac{\left(  2\pi\right)  ^{-1}}{\mathbb{P}\left(  Q_{1}\right)  }\right]
+\log h\left(  \omega\right)  -\log\left[  \sigma_{1}\sigma_{2}\sqrt
{1-\rho^{2}}A\right]  .
\end{equation}

$\left(  i\right)  $ \ We consider the region $-\omega_{0}>>1$ which, for
$\sigma_{2},$ $\mu_{1}>0,$ $\mu_{2}>0$ and $\left\vert \rho\right\vert <1$
fixed, means $\sigma_{1}$ is small. Using the standard bounds \cite{AB} for
the error function we write, for $\omega<0,$%
\begin{equation}
\log\left[  2\sqrt{\pi}\left(  -\omega_{0}\right)  e^{\omega_{0}^{2}}\right]
\leq\log h\left(  \omega_{0}\right)  \leq\log\left[  2\sqrt{\pi}\left(
-\omega_{0}\right)  e^{\omega_{0}^{2}}+\frac{1}{2\omega_{0}^{2}}\right]  \ .
\end{equation}
Approximating $\omega$ with $\omega_{0}$ and utilizing the bounds in the
proofs of the theorems above, \ as well as the $-C/2$ identity we write,%
\begin{align}
\log f_{X_{1}/X_{2}}\left(  x\ |\ Q_{1}\right)   &  \tilde{=}-2\log
x+2\rho\frac{\sigma_{1}}{\sigma_{2}}-\log\left(  1-\rho^{2}\right)
^{-1/2}+\log\left(  \frac{\sqrt{\pi}}{\mathbb{P}\left(  Q_{1}\right)  }\right)
\nonumber\\
&  -\frac{1}{2}\frac{\mu_{2}^{2}}{\sigma_{2}^{2}}-\log\left(  \frac{\sigma
_{2}}{\sigma_{1}}\right)  +\log\left(  -\omega_{0}\right)  .
\end{align}
Note that the difference between the two sides vanishes as $x\rightarrow
\infty$. All but the last two terms are clearly bounded for small $\sigma_{1}%
$. The last two are%
\begin{equation}
-\log\left(  \frac{\sigma_{2}}{\sigma_{1}}\right)  +\log\left(  -\omega
_{0}\right)  =\log\left(  \frac{\sigma_{1}}{\sigma_{2}}\left(  -\omega
_{0}\right)  \right)
\end{equation}%
\begin{equation}
=\log\left\{  \left[  2\left(  1-\rho^{2}\right)  \right]  ^{-1/2}\left(
\frac{\mu_{1}}{\sigma_{2}}-\rho\frac{\mu_{2}\sigma_{1}}{\sigma_{2}^{2}%
}\right)  \right\}  ,
\end{equation}
so that the sum of these two are also bounded for small $\sigma_{1}$. One can
then write%
\begin{equation}
\log f_{X_{1}/X_{2}}\left(  x\ |\ Q_{1}\right)  \tilde{=}-2\log x+2\rho
\frac{\sigma_{1}}{\sigma_{2}}+\log\left(  \sqrt{\frac{2}{\pi}}\frac
{1}{\mathbb{P}\left(  Q_{1}\right)  }\right)  +\log\left(  \frac{\mu_{1}%
}{\sigma_{2}}-\rho\frac{\mu_{2}\sigma_{1}}{\sigma_{2}^{2}}\right)
\end{equation}
and one has, provided $\sigma_{1}$ approaches zero more rapidly than $\left(
\log x\right)  ^{-1},$ the limit,%
\begin{equation}
\lim_{\substack{x\rightarrow\infty\\\sigma_{1}\rightarrow0}}\frac{\log
f_{X_{1}/X_{2}}\left(  x\ |\ Q_{1}\right)  }{\log x}=-2.
\end{equation}

$\bigskip$

$\left(  ii\right)  $ Next, we consider the $\omega_{0}>>1$ region, which, for
$\sigma_{1},$ $\mu_{1}>0,$ $\mu_{2}>0$ and $\left\vert \rho\right\vert <1$
fixed, implies $\sigma_{2}$ is small. This is the limit in which the
denominator becomes deterministic, as discussed above. We now attain this
limit from the $\log f_{X_{1}/X_{2}}\left(  x\ |\ Q_{1}\right)  $ expression.
From the identities of $C$ and $\omega_{0}$ above, we have%
\begin{align}
-\frac{C}{2}  &  =-\frac{1}{2}\left(  \frac{\mu_{2}}{\sigma_{2}}\right)
^{2}\ -\left(  \frac{1}{2\left(  1-\rho^{2}\right)  }\right)  \left\{
\rho^{2}\left(  \frac{\mu_{2}}{\sigma_{2}}\right)  ^{2}-2\rho\frac{\mu_{1}%
\mu_{2}}{\sigma_{1}\sigma_{2}}+\frac{\mu_{1}^{2}}{\sigma_{1}^{2}}\ \right\}
\nonumber\\
&  =-\frac{1}{2}\frac{1}{1-\rho^{2}}\left(  \frac{\mu_{2}}{\sigma_{2}}\right)
^{2}+\frac{\rho}{\left(  1-\rho^{2}\right)  }\frac{\mu_{1}\mu_{2}}{\sigma
_{1}\sigma_{2}}-\left(  \frac{1}{2\left(  1-\rho^{2}\right)  }\right)
\frac{\mu_{1}^{2}}{\sigma_{1}^{2}}%
\end{align}
With this expression and $\log h\left(  \omega_{0}\right)  \tilde{=}0,$ the
identity (\ref{logFromThm4.1}) yields
\[
\log f_{X_{1}/X_{2}}\left(  x\ |\ Q_{1}\right)  \tilde{=}-2\log x+2\rho
\frac{\sigma_{1}}{\sigma_{2}}+\log\left(  \frac{\left(  2\pi\right)  ^{-1}%
}{\mathbb{P}\left(  Q_{1}\right)  }\right)  +\log h\left(  \omega_{0}\right)
-\log\left(  \frac{\sigma_{2}}{\sigma_{1}}\frac{1}{\sqrt{1-\rho^{2}}}\right)
\]%
\begin{equation}
-\frac{1}{2}\frac{1}{1-\rho^{2}}\left(  \frac{\mu_{2}}{\sigma_{2}}\right)
^{2}+\frac{\rho}{\left(  1-\rho^{2}\right)  }\frac{\mu_{1}\mu_{2}}{\sigma
_{1}\sigma_{2}}-\left(  \frac{1}{2\left(  1-\rho^{2}\right)  }\right)
\frac{\mu_{1}^{2}}{\sigma_{1}^{2}}%
\end{equation}
As $\sigma_{2}\rightarrow0$ the right hand side of this expression diverges as
$\sigma_{2}^{-2}$. This means that the exponent of the decay rate of
$X_{1}/X_{2}$ diverges to $-\infty$. This is consistent with the previous
result and the expecation that when $\sigma_{2}\rightarrow0$ the denominator
approaches a constant, so that the only randomness is in the Gaussian numerator.

\textbf{\bigskip}

\section{Appendix A: Calculation of $\mathbb{P}\left(  Q_{1}\right)
,...,\mathbb{P}\left(  Q_{4}\right)  $}

We label the quadrants $Q_{1},...,Q_{4}$ in the usual way. We define the
multivariable distribution function $F\left(  x_{1},x_{2}\right)  $ in terms
of the density as%
\begin{equation}
F\left(  x_{1},x_{2}\right)  :=\int_{-\infty}^{0}ds_{1}\int_{-\infty}%
^{0}ds_{2}f\left(  s_{1},s_{2}\right)  .
\end{equation}
We would like to evaluate the probability of $\left(  X_{1},X_{2}\right)  $
being in \ the first quadrant, i.e., in $Q_{1}$ which is%
\begin{equation}
\mathbb{P}\left(  Q_{1}\right)  =\int_{0}^{\infty}ds_{1}\int_{0}^{\infty
}ds_{2}f\left(  s_{1},s_{2}\right)  .
\end{equation}
A basic decomposition of the $\mathbb{R}^{2}$ integrals yields
\begin{equation}
\mathbb{P}\left(  Q_{1}\right)  =1-F\left(  0,\infty\right)  -F\left(
\infty,0\right)  +F\left(  0,0\right)  .
\end{equation}
We now evaluate the right hand side. Recall that $\rho\in\left(  -1,1\right)
.$ Let $\delta_{\rho}:=sgn\rho$ and
\begin{equation}
a_{i}:=\left(  x_{i}-\mu_{i}\right)  /\sigma_{i}%
\end{equation}
so $a_{i}=\infty$ if $x_{i}=\infty$ and $a_{i}=-\mu_{i}/\sigma_{i}$ if
$x_{i}=0.$

We have, using using $\left(  2.2.3\right)  $ in \cite{TO} and $\Phi\left(
\infty\right)  =1,$ the identities
\begin{equation}
F\left(  x_{1},x_{2}\right)  =\int_{-\infty}^{\infty}\Phi\left(  \frac
{\sqrt{\left\vert \rho\right\vert }z+a_{1}}{\sqrt{1-\left\vert \rho\right\vert
}}\right)  \Phi\left(  \frac{\delta_{\rho}\sqrt{\left\vert \rho\right\vert
}z+a_{2}}{\sqrt{1-\left\vert \rho\right\vert }}\right)  \phi\left(  z\right)
dz,
\end{equation}%
\begin{align}
F\left(  0,0\right)   &  =\int_{-\infty}^{\infty}\Phi\left(  \frac
{\sqrt{\left\vert \rho\right\vert }z-\mu_{1}/\sigma_{1}}{\sqrt{1-\left\vert
\rho\right\vert }}\right)  \Phi\left(  \frac{\delta_{\rho}\sqrt{\left\vert
\rho\right\vert }z-\mu_{2}/\sigma_{2}}{\sqrt{1-\left\vert \rho\right\vert }%
}\right)  \phi\left(  z\right)  dz,\nonumber\\
F\left(  \infty,0\right)   &  =\int_{-\infty}^{\infty}\Phi\left(  \frac
{\delta_{\rho}\sqrt{\left\vert \rho\right\vert }z-\mu_{2}/\sigma_{2}}%
{\sqrt{1-\left\vert \rho\right\vert }}\right)  \phi\left(  z\right)  dz,\\
F\left(  0,\infty\right)   &  =\int_{-\infty}^{\infty}\Phi\left(  \frac
{\sqrt{\left\vert \rho\right\vert }z-\mu_{1}/\sigma_{1}}{\sqrt{1-\left\vert
\rho\right\vert }}\right)  \phi\left(  z\right)  dz.\nonumber
\end{align}
Combining these with the above expression for $\mathbb{P}\left(  Q_{1}\right)
$ we have the calculation in closed form.

Note that one can also modify the expression (2.2.3) in \cite{TO} to calculate
this in another way by deriving the analogous relation for $\tilde{F}\left(
x_{1},x_{2}\right)  :=\int_{x_{1}}^{\infty}ds_{1}\int_{x_{2}}^{\infty}%
ds_{2}f\left(  s_{1},s_{2}\right)  $ so $\mathbb{P}\left\{  X_{1}\geq
0,\ X_{2}\geq0\right\}  =\tilde{F}\left(  0,0\right)  $ .

The calculations for $\mathbb{P}\left(  Q_{2}\right)  ,$ $\mathbb{P}\left(
Q_{3}\right)  ,$ $\mathbb{P}\left(  Q_{4}\right)  $ are similar.

Defining $H_{T}:=\left\{  X_{2}>0\right\}  $ and $H_{B}:=\left\{
X_{2}<0\right\}  $ we note%

\begin{equation}
\mathbb{P}\left(  H_{B}\right)  =F\left(  \infty,0\right)  ,\ \mathbb{P}%
\left(  H_{T}\right)  =1-F\left(  \infty,0\right)  .
\end{equation}

\bigskip

\section{Appendix B:\ Other Quadrants}

We have been considering $X_{1}/X_{2}$ when both $X_{1}$ and $X_{2}$ are
positive. The remaining possibilities are calculated below. We divide $\left(
X_{1},X_{2}\right)  $ space into the usual quadrants $Q_{1},Q_{2},Q_{3}$ and
$Q_{4}$ and also let $H_{T}$ and $H_{B}$ denote the half-spaces consisting of
$X_{2}>0$ and $X_{2}<0$ respectively. Note that sets such as $\left\{
X_{2}=0\right\}  $ will be of measure zero in terms of the multivariate
density, which is written in terms of the exponential $\phi$ defined earlier
as
\begin{equation}
f\left(  s_{1},s_{2}\right)  =\frac{1}{\sigma_{1}\sigma_{2}\sqrt{1-\rho^{2}}%
}\phi\left(  \frac{s_{2}-\mu_{2}}{\sigma_{2}}\right)  \phi\left(  \frac
{s_{1}-\left(  \mu_{1}+\rho\frac{\sigma_{1}}{\sigma_{2}}\left(  s_{2}-\mu
_{2}\right)  \right)  }{\sigma_{1}\sqrt{1-\rho^{2}}}\right)  .
\end{equation}

\subsection{Obtaining the density without conditioning.}

$\left(  i\right)  $ Upon using the half-spaces we can write%
\begin{align}
\mathbb{P}\left\{  X_{1}/X_{2}\leq x\ and\ X_{2}>0\right\}   &  =\int
_{0}^{\infty}ds_{2}\int_{-\infty}^{xs_{2}}ds_{1}f\left(  s_{1},s_{2}\right)
\nonumber\\
\ \mathbb{P}\left\{  X_{1}/X_{2}\leq x\ and\ X_{2}<0\right\}   &
=\int_{-\infty}^{0}ds_{2}\int_{xs_{2}}^{\infty}ds_{1}f\left(  s_{1}%
,s_{2}\right)  .
\end{align}
Using the theorem of total probability and ignoring sets of measure \ zero we
write%
\begin{align}
\mathbb{P}\left\{  X_{1}/X_{2}\leq x\ \right\}   &  =\mathbb{P}\left\{
X_{1}/X_{2}\leq x\ and\ X_{2}>0\right\} \nonumber\\
&  +\mathbb{P}\left\{  X_{1}/X_{2}\leq x\ and\ X_{2}<0\right\}  \ .
\end{align}
Differentiating this expression, we obtain the density (without conditioning)
\[
f_{X_{1}/X_{2}}\left(  x\right)  =\partial_{x}\mathbb{P}\left\{  X_{1}%
/X_{2}\leq x\right\}  =\int_{0}^{\infty}f\left(  xs_{2},s_{2}\right)
s_{2}ds_{2}-\int_{-\infty}^{0}s_{2}f\left(  xs_{2},s_{2}\right)  ds_{2}.
\]
The first integral has already been calculated. Using the symmetry of $\phi$
we can rewrite the second integral as%
\begin{gather}
-\int_{-\infty}^{0}s_{2}f\left(  xs_{2},s_{2}\right)  =\nonumber\\
\frac{1}{\sigma_{1}\sigma_{2}\sqrt{1-\rho^{2}}}\int_{0}^{\infty}z_{2}%
\phi\left(  \frac{z_{2}-\left(  -\mu_{2}\right)  }{\sigma_{2}}\right)
\phi\left(  \frac{xz_{2}-\left(  -\mu_{1}+\rho\frac{\sigma_{1}}{\sigma_{2}%
}\left[  z_{2}-\left(  -\mu_{2}\right)  \right]  \right)  }{\sigma_{1}%
\sqrt{1-\rho^{2}}}\right)  .
\end{gather}
We see that this is the same integral as the first with $\left(  \mu_{1}%
,\mu_{2}\right)  $ replaced by $\left(  -\mu_{1},-\mu_{2}\right)  ,$ and thus
has similar properties.

Thus we can express the (non-conditioned) density as%
\begin{gather}
f_{X_{1}/X_{2}}\left(  x\right)  =\partial_{x}\mathbb{P}\left\{  X_{1}%
/X_{2}\leq x\right\}  =\nonumber\\
\int_{0}^{\infty}s_{2}ds_{2}\frac{1}{\sigma_{1}\sigma_{2}\sqrt{1-\rho^{2}}%
}\phi\left(  \frac{s_{2}-\mu_{2}}{\sigma_{2}}\right)  \phi\left(  \frac
{xs_{2}-\left(  \mu_{1}+\frac{\rho\sigma_{1}}{\sigma_{2}}\left(  s_{2}-\mu
_{2}\right)  \right)  }{\sigma_{1}\sqrt{1-\rho^{2}}}\right) \nonumber\\
+\frac{1}{\sigma_{1}\sigma_{2}\sqrt{1-\rho^{2}}}\int_{0}^{\infty}z_{2}%
\phi\left(  \frac{z_{2}-\left(  -\mu_{2}\right)  }{\sigma_{2}}\right)
\phi\left(  \frac{xz_{2}-\left(  -\mu_{1}+\rho\frac{\sigma_{1}}{\sigma_{2}%
}\left[  z_{2}-\left(  -\mu_{2}\right)  \right]  \right)  }{\sigma_{1}%
\sqrt{1-\rho^{2}}}\right)  \ .
\end{gather}

Hence, this can be analyzed asymptotically in the same manner as the
conditional probability with similar results.

$\left(  ii\right)  $ We can compute the other conditional probabilities in
addition to $f_{X_{1}/X_{2}}\left(  x\ |\ Q_{1}\right)  .$ For $X_{1}<0,$
$X_{2}<0$ \ we have%
\begin{align}
&  \mathbb{P}\left\{  X_{1}/X_{2}\leq x\ and\ X_{1}<0,\ X_{2}<0\right\}
\nonumber\\
&  =\mathbb{P}\left\{  X_{1}\geq xX_{2}\ and\ X_{1}<0,\ X_{2}<0\right\}
\nonumber\\
&  =\left\{
\begin{array}
[c]{c}%
\int_{-\infty}^{0}ds_{2}\int_{xs_{2}}^{0}ds_{1}f\left(  s_{1},s_{2}\right) \\
0
\end{array}
\right.
\begin{array}
[c]{c}%
if\ \ x>0\\
if\text{ }x\leq0
\end{array}
\ \ ,
\end{align}
Hence, we can differentiate with respect to $x$ and use the definition of
conditional density to obtain, for $x\geq0,$
\begin{align}
f_{X_{1}/X_{2}}\left(  x\ |\ Q_{3}\right)   &  =\partial_{x}\mathbb{P}\left(
X_{1}/X_{2}\leq x\ |\ Q_{3}\right) \nonumber\\
&  =\left[  \mathbb{P}\left(  Q_{3}\right)  \right]  ^{-1}\partial_{x}%
\int_{-\infty}^{0}ds_{2}\int_{xs_{2}}^{0}ds_{1}f\left(  s_{1},s_{2}\right)
\nonumber\\
&  =-\left[  \mathbb{P}\left(  Q_{3}\right)  \right]  ^{-1}\int_{-\infty}%
^{0}s_{2}f\left(  xs_{2},s_{2}\right)  ds_{2},
\end{align}
and \ $f_{X_{1}/X_{2}}\left(  x\ |\ Q_{3}\right)  =0$ if $x<0.$

Similarly, we write%
\begin{equation}
f_{X_{1}/X_{2}}\left(  x\ |\ Q_{2}\right)  =\left\{
\begin{array}
[c]{ccc}%
\left[  \mathbb{P}\left(  Q_{2}\right)  \right]  ^{-1}\int_{0}^{\infty
}f\left(  xs_{2},s_{2}\right)  ds_{2} & if & x<0\\
0 & if & x\geq0
\end{array}
\right.
\end{equation}
and
\begin{equation}
f_{X_{1}/X_{2}}\left(  x\ |\ Q_{4}\right)  =\left\{
\begin{array}
[c]{ccc}%
-\left[  \mathbb{P}\left(  Q_{4}\right)  \right]  ^{-1}\int_{-\infty}%
^{0}f\left(  x_{s_{2}},s_{2}\right)  ds_{2} & if & x<0\\
0 & if & x\geq0
\end{array}
\right.
\end{equation}

\bigskip

\section{Appendix C: Relations between conditional densities.}

For any $x\in\mathbb{R}$ we have%
\begin{align}
\mathbb{P}\left\{  X_{1}/X_{2}\leq x,\ X_{2}>0\right\}   &  =\mathbb{P}%
\left\{  X_{1}/X_{2}\leq x,\ X_{2}>0,X_{1}>0\right\} \nonumber\\
&  +\mathbb{P}\left\{  X_{1}/X_{2}\leq x,\ X_{2}>0,X_{1}<0\right\}  \ .
\end{align}
Let $H_{R}:=\left\{  X_{1}>0,X_{2}\in\mathbb{R}\right\}  $ and $H_{L}%
:=\left\{  X_{1}<0,X_{2}\in\mathbb{R}\right\}  $.

For $x>0$ the probability that $X_{1}/X_{2}\leq x$ is $1$ if $X_{2}>0$ and
$X_{1}<0,$ so differentiating the term above yields%
\begin{equation}
\partial_{x}\mathbb{P}\left\{  X_{1}/X_{2}\leq x,\ X_{2}>0\right\}
=\partial_{x}\mathbb{P}\left\{  X_{1}/X_{2}\leq x,\ X_{2}>0,\ X_{1}>0\right\}
\ .
\end{equation}
Rewriting each side using conditional probability, we have,%
\begin{align}
&  \partial_{x}\mathbb{P}\left\{  X_{1}/X_{2}\leq x\ |\ X_{2}>0\right\}
\mathbb{P}\left\{  X_{2}>0\right\} \nonumber\\
&  =\partial_{x}\mathbb{P}\left\{  X_{1}/X_{2}\leq x\ |\ X_{2}>0,\ X_{1}%
>0\right\}  \mathbb{P}\left(  Q_{1}\right)
\end{align}
sand we have, in terms of conditional probability,%
\begin{align}
&  \partial_{x}\mathbb{P}\left\{  X_{1}/X_{2}\leq x\ |\ X_{2}>0\right\}
\mathbb{P}\left(  H_{T}\right) \nonumber\\
&  =\partial_{x}\mathbb{P}\left\{  X_{1}/X_{2}\leq x\ |\ X_{2}>0,\ X_{1}%
>0\right\}  \mathbb{P}\left(  Q_{1}\right)  .
\end{align}
The differentiated terms are just conditional densities, and we write
\begin{equation}
f_{X_{1}/X_{2}}\left(  x\ |\ H_{T}\right)  \mathbb{P}\left(  H_{T}\right)
=f_{X_{1}/X_{2}}\left(  x\ |Q_{1}\right)  \mathbb{P}\left(  Q_{1}\right)  ~.
\end{equation}

For $x<0$ the probability that $X_{1}/X_{2}\leq x$ is $0$ if $X_{1}>0$ and
$X_{2}>0.$ Thus we can write, from the first expression,%
\begin{equation}
\partial_{x}\mathbb{P}\left\{  X_{1}/X_{2}\leq x,\ X_{2}>0\right\}
=\partial_{x}\mathbb{P}\left\{  X_{1}/X_{2}\leq x,\ X_{2}>0,X_{1}<0\right\}
\ .
\end{equation}
Using the same procedure as above, we write%
\begin{align}
&  \partial_{x}\mathbb{P}\left\{  X_{1}/X_{2}\leq x\ |X_{2}>0\right\}
\mathbb{P}\left\{  X_{2}>0\right\} \nonumber\\
&  =\partial_{x}\mathbb{P}\left\{  X_{1}/X_{2}\leq x\ |\ X_{2}>0,X_{1}%
<0\right\}  \mathbb{P}\left(  Q_{2}\right)  \ .
\end{align}
We can then write the expression in terms of conditional density as
\begin{equation}
f_{X_{1}/X_{2}}\left(  x\ |\ H_{R}\right)  \mathbb{P}\left(  H_{R}\right)
=f_{X_{1}/X_{2}}\left(  x\ |\ Q_{2}\right)  \mathbb{P}\left(  Q_{2}\right)  .
\end{equation}

\end{document}